\documentclass[11pt,a4paper]{article}
\usepackage{amsmath,amssymb,amsthm}
\usepackage{mathtools}
\usepackage{geometry}
\usepackage{enumitem}
\usepackage{microtype}
\usepackage{comment}
\usepackage{cite}
\newtheorem{theorem}{Theorem}
\newtheorem{proposition}{Proposition}
\newtheorem{corollary}{Corollary}
\newtheorem{remark}{Remark}

\newcommand{\calC}{\mathcal{C}}
\newcommand{\calT}{\mathcal{T}}
\newcommand{\calX}{\mathcal{X}}
\newcommand{\calY}{\mathcal{Y}}
\newcommand{\calH}{\mathcal{H}}
\newcommand{\EFA}{E_{\mathrm{FA}}}
\newcommand{\EMD}{E_{\mathrm{MD}}}

\newcommand{\E}{\mathbb{E}}
\newcommand {\dfn} {\stackrel{\Delta} {=}}

\title{Soft Covering via Hypothesis Testing:\\
Typical-Code Exponents and Mismatched Detection}

\author{Neri~Merhav\\
The Viterbi Faculty of Electrical and Computer Engineering\\
Technion -- Israel Institute of Technology\\
Technion City, Haifa 3200003, Israel\\
E-mail: {\tt merhav@technion.ac.il}}

\date{\today}

\begin{document}
\maketitle
\thispagestyle{empty}

\begin{abstract}
We study the typical-code (quenched) behavior of the false-alarm (FA) and
missed-detection (MD) error exponents of the Neyman--Pearson
test associated with soft covering,
complementing the average-code (annealed) analysis that has been
carried out in a companion paper~\cite{MerhavCompanion}.
We prove that, as the block-length tends to infinity, for almost every
randomly selected fixed-composition codebook, the negative normalized
logarithms of both error probabilities
converge to their respective average-code exponents. In other words,
the error exponents are \emph{self-averaging}.
We then extend the scope and study a \emph{mismatched} likelihood ratio test
that assumes the wrong channel
model. Here, we derive the mismatched error exponents, show that self-averaging
persists under mismatch, and characterize the degradation.
In particular, we characterize the coding rate beyond which the two kinds of
error exponents cannot be both positive at the same time, which in the matched
case, is given by the channel input-output mutual information rate.

\end{abstract}

\vspace{1.5\baselineskip}
\setlength{\baselineskip}{1.5\baselineskip}

\section{Introduction}
\label{sec:intro}

The soft covering lemma~\cite{Wyner75,HanVerdu93,Cuff13,CuffStrong15,Cuff16}
asserts that a random codebook of rate $R$ larger than the channel
input-output per-letter mutual information, $I(X;Y)$, causes
the output mixture distribution $P_{Y^n|\calC}$ to be nearly indistinguishable
from the i.i.d.\ distribution, $P_Y^{\otimes n}$, induced by the input
distribution and the channel.
Introduced by Wyner~\cite{Wyner75} as a tool for proving
the common information theorem and systematically developed
by Han and Verd\'u~\cite{HanVerdu93} under the name of channel
resolvability, soft covering underpins wiretap secrecy,
identification coding, common randomness generation, and
channel synthesis~\cite{Cuff13}.

A substantial body of work characterizes the \emph{exact exponent}
at which the soft covering approximation improves for $R>I(X;Y)$,
under various distance measures.
Hayashi~\cite{Hayashi06} obtained a lower bound on the exponent
under the Kullback-Leibler (KL) divergence.
Parizi, Telatar and Merhav~\cite{PTM17} derived the exact exponent
under KL divergence with application to the wiretap channel.
Yu and Tan~\cite{YuTan19} characterized the exact exponent under
R\'enyi divergence of order $\alpha\in(0,2)$.
Yagli and Cuff~\cite{YagliCuff19} derived the exact exponent under
total variation distance.
Li, Li and Yu~\cite{LLY26} recently derived a strong-converse
exponent under KL divergence.
Cuff~\cite{CuffStrong15,Cuff16} showed that soft covering holds
with doubly exponential probability over the random codebook,
enabling applications via the union bound.

In a companion paper~\cite{MerhavCompanion}, we studied soft covering
through the lens of Neyman--Pearson hypothesis testing and derived
the false-alarm (FA) and missed-detection
(MD) error exponents,
$\EFA^{\mbox{\tiny a}}(\tau,R)$ and $\EMD^{\mbox{\tiny a}}(\tau,R)$,
as functions of the coding rate $R$ and the log-likelihood ratio (LLR) threshold $\tau$.
The analysis in \cite{MerhavCompanion} was carried out for the \emph{average
code}, namely, it concerns the long-block limits of the normalized logarithms of the expected
error probabilities of both kinds, where the expectation is with respect to
(w.r.t.) the random selection of a fixed-composition code. In the jargon of
statistical physics, these exponents can be referred to as \emph{annealed} error exponents, 
in analogy to annealed free energies (and other physical quantities)
associated with disordered physical systems.
The analysis in \cite{MerhavCompanion} revealed a rich phase structure and characterized
the soft-covering phenomenon in Neyman--Pearson terms.

A natural follow-up question that may arise
concerns the characterization of \emph{quenched} error exponents, namely, the long-block limits of the normalized
expected logarithms (as opposed to logarithms of expectations) of the FA and
MD probabilities, and a further question is whether they also reflect the almost-sure limits, in
which case, they can be referred to as \emph{typical-code} error exponents,
and we say that the \emph{self-averaging} property (which is actually an
ergodic property) holds true.
Obviously, by Jensen's inequality, the quenched FA and MD error exponents
cannot be smaller than their respective annealed counterparts, but can we
characterize them beyond that? 

We show in this work that the quenched FA and MD error
exponents are equal to the respective annealed error exponents of
\cite{MerhavCompanion} for all $R$ and $\tau$, and they also admit self-averaging.
This equivalence between quenched and annealed exponents is quite
surprising in view of the fact that in other related problem areas in information
theory, there are gaps between annealed and quenched quantities. Random coding
error exponents of channel coding (see, e.g., \cite{Gallager65,Gallager68}), which are traditionally defined in the
annealed sense, are smaller in general than their typical-code counterparts
(see, e.g., \cite{BargForney02}, \cite{merhavtypical} and references therein) at a certain range of low coding rates.
In statistical physics, the analogous
distinction between quenched and annealed free energies
of disordered systems with random parameters
is fundamental~\cite{MezardParisi87}, \cite{MM09}, and the gaps are visible, especially at
low temperatures, and they convey different meanings. In Section 
\ref{sec:jensen}, an attempt is made to provide some insight and intuition
behind the equivalence between quenched and annealed exponents in
our problem, as opposed to the inequivalence in the other problem areas.

We then proceed to study the effect of mismatched detection in the context of
soft covering. Suppose that 
the detector implements the likelihood ratio test (LRT) w.r.t.\ the wrong
channel model, say $\tilde{W}$, instead of the
true channel $W$. As in mismatched decoding (see, e.g.,
\cite{SFSM20} and many references therein),
this mismatch may arise from either a wrong assumption on
the channel statistics, or due to limitations on the implementation of
the exact LRT.
We show that self-averaging persists under mismatch,
characterize the degradation of the exponents,
and identify a precise connection between the shift
in the soft-covering critical rate and the generalized mutual information (GMI) of
mismatched decoding.

The outline of this article is as follows. 
Section~\ref{sec:setup} recalls the setup and notation.
Sections~\ref{sec:FA} and~\ref{sec:MD} prove the main self-averaging theorems.
Section~\ref{sec:jensen} discusses why the Jensen gap vanishes.
Section~\ref{sec:mismatch} develops the mismatched detection theory,
including the annealed mismatched exponents, self-averaging under mismatch,
the phase structure, and a connection to the GMI.
Finally, in Section~\ref{sec:conclusion} we summarize and conclude.

\section{Setup, Notation Conventions and Objectives}
\label{sec:setup}

We adopt the same notation conventions as in \cite{MerhavCompanion}.
Let $W:\calX\to\calY$ be a discrete memoryless channel (DMC) with finite input
and output alphabets, $\calX$ and $\calY$,
and let $P_X$ be a fixed input distribution. We denote by $P_Y$ the output marginal
induced by $P_X$ and $W$, that is,
\begin{equation}
P_Y(y)=\sum_{x\in\calX} P_X(x)W(y|x),~~~~\forall~y\in\calY.
\end{equation}
For a given positive integer $n$ and a sequence $x^n=(x_1,\ldots,x_n)\in\calX^n$,
the channel output distribution is given by
\begin{equation}
W^n(y^n|x^n)=\prod_{i=1}^n W(y_i|x_i),~~~~~y^n\in\calY^n,
\end{equation}
where $\calX^n$ and $\calY^n$ are the $n$-th Cartesian powers of $\calX$ and
$\calY$, respectively.
We also denote 
\begin{equation}
P_Y^{\otimes n}(y^n)=\prod_{i=1}^n P_Y(y_i).
\end{equation}
The \emph{type} (empirical distribution) of $x^n$ is
defined as $\hat{P}_{x^n}(a)=\frac{1}{n}\#\{i:x_i=a\}$, $a\in\calX$.
The \emph{type class} $\calT(Q_X)\subseteq\calX^n$ is the set
of all sequences of type $\hat{P}_{x^n}(a)=Q_X$.
For a pair $(x^n,y^n)$, the \emph{joint type} is
$Q_{XY}(a,b)=\frac{1}{n}\#\{i:x_i=a,~y_i=b\}$, $(a,b)\in\calX\times\calY$,
with marginals $Q_X$, $Q_Y$ and conditional type $Q_{Y|X}$.
The \emph{conditional type class} is defined as
$\calT(Q_{X|Y}|y^n)=\{x^n:~(x^n,y^n)\in\calT(Q_{XY})\}$.
Likewise, $\calT(Q_{Y|X}|x^n)\dfn\{y^n:~(x^n,y^n)\in\calT(Q_{XY})\}$.
We always restrict to types with $Q_X=P_X$.

Information functionals derived from a given probability distribution
will be subscripted by the notation of this distribution. When this
is an empirical distribution $Q_{XY}$, the subscript will be
abbreviated by $Q$, in order to avoid cumbersome notation. Thus,
$H_Q(Y)=-\sum_y Q_Y(y)\log Q_Y(y)$ is the marginal empirical entropy of an
auxiliary random vector $Y$ governed by $Q_Y$,
$H_Q(Y|X)=-\sum_{x,y}Q_{XY}(x,y)\log Q_{Y|X}(y|x)$ is the
conditional empirical entropy of $Y$ given $X$, where $(X,Y)$ are jointly governed by $Q_{XY}$, and
$I_Q(X;Y)=H_Q(Y)-H_Q(Y|X)$ is the mutual information under $Q_{XY}$.

The channel codebook $\calC=\{x^n(1),\ldots,x^n(M)\}$,
$M=e^{nR}$, $R$ being the coding rate in nats per channel use, has its codewords drawn independently and uniformly
from the type class $\mathcal{T}(P_X)$.
For a given type $Q_{XY}$, $y^n\in\calY^n$, let $N(Q_{XY}|y^n)$ 
denote the number of codewords $\{x^n(m)\}$ in $\calC$ whose joint type with
$y^n$ is $Q_{XY}$.
Since the codewords are randomly selected
independently, $N(Q_{XY}|y^n)$ is a binomial random variable with $M=e^{nR}$ trials and
success rate of the exponential order of $e^{-nI_{Q}(X;Y)}$.
When $I_Q(X;Y)\leq R$, there are typically about $e^{n[R-I_Q(X;Y)]}\to\infty$ codewords in
$\calT(Q_{X|Y}|y^n)$. We therefore
refer to $Q_{XY}$ as a \emph{bulk type} whenever $I_Q(X;Y)\leq R$.
On the other hand, when $R\le I_Q(X;Y)$ the probability that $N(Q_{XY}|y^n)\ge
1$ tends to zero at the exponential rate of $e^{n[I_Q(X;Y)-R]}$, and so, there are typically no codewords at
all in $\calT(Q_{X|Y}|y^n)$, and then $Q_{XY}$ is referred to as a
\emph{sparse type}.

The hypothesis testing problem pertaining to soft covering is in distinguishing between the following
two hypotheses concerning the probability distribution that governs a vector of observations, $y^n\in\calY^n$:\\
$\calH_0$: $y^n$ is governed by the i.i.d.\ distribution, $P_Y^{\otimes n}$.\\
$\calH_1$: $y^n$ is governed by the mixture disribution
$\frac{1}{M}\sum_{m=1}^M W^n(y^n|x^n(m))$.

Defining
\begin{equation}
S(y^n,\calC)=\sum_{m=1}^M W^n(y^n|x^n(m)),
\end{equation}
the LLR statistic is defined as
\begin{equation}
	\Lambda(y^n,\calC)=
	=\frac{1}{n}\log\frac{\frac{1}{M}\sum_{m=1}^MW^n(y^n|x^n)}{P_Y^{\otimes n}(y^n)}
=\frac{1}{n}\log\frac{S(y^n,\calC)}{M\cdot P_Y^{\otimes n}(y^n)}.
\end{equation}
For a given threshold parameter $\tau$, if $\Lambda(y^n,\calC)\ge\tau$, hypothesis $\calH_1$ is accepted, otherwise,
hypothesis $\calH_0$ is accepted. 
The FA and MD probabilities (for a given $\calC$ known to the detector) are then defined as
\begin{align}
  \alpha_n(\tau,R,\calC)
  &\dfn\sum_{\{y^n:~\Lambda(y^n,\calC)\geq\tau\}}P_Y^{\otimes n}(y^n)\\
  \beta_n(\tau,R,\calC)
  &\dfn\frac{1}{M}\sum_{m=1}^M\sum_{\{y^n:~\Lambda(y^n,\calC)<\tau\}}
    W^n(y^n|x^n(m)).
\end{align}
The corresponding annealed (average-code) exponents are defined as
\begin{eqnarray}
\EFA^{\mbox{\tiny a}}(\tau,R)&\dfn&-\lim_{n\to\infty}\frac{1}{n}\log
\E\{\alpha_n(\tau,R,\calC)\},
\label{eq:ann_FA_def}\\
\EMD^{\mbox{\tiny a}}(\tau,R)&\dfn&-\lim_{n\to\infty}\frac{1}{n}\log\E\{\beta_n(\tau,R,\calC)\}
\label{eq:ann_MD_def}
\end{eqnarray}
where the expectations are w.r.t.\ the random selection of $\calC$, and where
the existence of these limits was shown in \cite{MerhavCompanion}.
For a joint type $Q_{XY}$ with $Q_X=P_X$, define:
\begin{align}
  D_{\mbox{\tiny m}}(Q_Y)
  &\dfn D(Q_Y\|P_Y),\\
  D_{\mbox{\tiny c}}(Q_{XY})
  &\dfn D(Q_{Y|X}\|W|P_X)
   = \sum_{x,y}P_X(x)Q_{Y|X}(y|x)\log\frac{Q_{Y|X}(y|x)}{W(y|x)},
\end{align}
and observe that $D_{\mbox{\tiny m}}(Q_Y)\leq D_{\mbox{\tiny c}}(Q_{XY})$
by the data processing inequality of the KL divergence (see also \cite{MerhavCompanion}).
We also define
\begin{equation}
\lambda(Q_{XY},R)
\dfn D_{\mbox{\tiny m}}(Q_Y)-D_{\mbox{\tiny c}}(Q_{XY})+[I_Q(X;Y)-R]_+,
\label{eq:lambda_def}
\end{equation}
and
\begin{eqnarray}
\Delta(Q_Y,R)
&\dfn&\max_{\substack{Q'_{XY}:\,Q'_X=P_X,\,Q'_Y=Q_Y\\
I_{Q'}(X;Y)\leq R}}
\bigl[D_{\mbox{\tiny m}}(Q_Y)-D_{\mbox{\tiny c}}(Q'_{XY})\bigr]\nonumber\\
&=&\max_{\substack{Q'_{XY}:\,Q'_X=P_X,\,Q'_Y=Q_Y\\
I_{Q'}(X;Y)\leq R}}
\lambda(Q'_{XY},R).
\label{eq:Delta_def}
\end{eqnarray}
In \cite{MerhavCompanion} the following single-letter expressions for the
annealed exponents $\EFA^{\mbox{\tiny a}}(\tau,R)$
and $\EMD^{\mbox{\tiny a}}(\tau,R)$ were derived:
\begin{eqnarray}
\EFA^{\mbox{\tiny a}}(\tau,R)&=&
\min_{\substack{Q_{XY}:\,Q_X=P_X\\\lambda(Q_{XY},R)\geq\tau}}
\bigl\{D_{\mbox{\tiny m}}(Q_Y)+[I_Q(X;Y)-R]_+\bigr\},
\label{eq:ann_FA}\\
\EMD^{\mbox{\tiny a}}(\tau,R)&=&
\min_{\substack{Q_{XY}:\,Q_X=P_X\\\lambda(Q_{XY},R)<\tau\\
\Delta(Q_Y,R)<\tau}}
D_{\mbox{\tiny c}}(Q_{XY}).
\label{eq:ann_MD}
\end{eqnarray}

Some more notation conventions that will be used throughout the sequel include the following. The logarithmic
function $\log(\cdot)$ will be understood to be defined to the base
$\mbox{e}$ unless specified otherwise. The notation $[\cdot]_+$ designates the
positive clipping operator, namely, given a real number $t$, $[t]_+$ is defined as
$\max\{t,0\}$. For two positive sequences, $\{a_n\}_{n\ge 1}$ and
$\{b_n\}_{n\ge 1}$, the notation $a_n\doteq b_n$ means equivalence in the
exponential scale, i.e., $\lim_{n\to\infty}\frac{1}{n}\log\frac{a_n}{b_n}=0$.

Our first objective in this work is to derive the quenched (typical-code)
FA and MD exponents, defined as
\begin{eqnarray}
\EFA^{\mbox{\tiny q}}(\tau,R)&\dfn&-\lim_{n\to\infty}\frac{1}{n}\E\left\{\log
\alpha_n(\tau,R,\calC)\right\},
\label{eq:que_FA_def}\\
\EMD^{\mbox{\tiny
q}}(\tau,R)&\dfn&-\lim_{n\to\infty}\frac{1}{n}\E\left\{\log\beta_n(\tau,R,\calC)\right\}
\label{eq:que_MD_def}
\end{eqnarray}
where it should be noted that here the exponents are defined in terms of
expected logarithms, as opposed to eqs.\ (\ref{eq:ann_FA_def}) and
(\ref{eq:ann_MD_def}), that are defined in terms of logarithms of expectations.
We will prove that $\EFA^{\mbox{\tiny q}}(\tau,R)=\EFA^{\mbox{\tiny a}}(\tau,R)$
and $\EMD^{\mbox{\tiny q}}(\tau,R)=\EMD^{\mbox{\tiny a}}(\tau,R)$,
and moreover that these common values are also the almost-sure limits of
$\left\{-\frac{1}{n}\log\alpha_n(\tau,R,\calC)\right\}$ and
$\left\{-\frac{1}{n}\log\beta_n(\tau,R,\calC)\right\}$, respectively ---
i.e., the exponents are self-averaging.

Our second objective will be to expand the scope to mismatched hypothesis
testing, where the detector assumes the wrong channel $\tilde{W}$ instead of
the true one, $W$. It turns out that the expressions of the error exponents extend relatively
straightforwardly to the mismatched case. Moreover, both the quenched-annealed
equivalence and the self-averaging properties are preserved. The more
interesting and less trivial part is the effect of mismatch on the critical rate $R$ beyond
which there is no value of the threshold $\tau$ where both error exponents are
strictly positive, which in the matched case is $I(X;Y)$ --- the soft-covering
limit. As mentioned before, it will turn out that there is a non-trivial relationship between the mismatched critical
rate and the GMI of mismatched decoding.

\section{The FA Exponent}
\label{sec:FA}

Our first theorem in this work asserts that 
$\left\{-\frac{1}{n}\log\alpha_n(\tau,R,\calC)\right\}_{n\ge 1}$
is self-averaging and its almost-sure limit is $\EFA^{\mbox{\tiny
a}}(\tau,R)$. It will then follow that
$\left\{-\frac{1}{n}\E\{\log\alpha_n(\tau,R,\calC)\}\right\}_{n\ge 1}$
converges to the same limit, which then must coincide with $\EFA^{\mbox{\tiny
q}}(\tau,R)$. 

\begin{theorem}
\label{thm:FA}
For all $\tau\in\mathbb{R}$ and $R>0$:
\begin{equation}
  \frac{1}{n}\log\alpha_n(\tau,R,\calC)
  \;\xrightarrow{n\to\infty}\;-\EFA^{\mbox{\tiny a}}(\tau,R)
  \quad\textup{a.s.\ over }\calC.
  \label{eq:FA_as}
\end{equation}
\end{theorem}

\begin{proof}
Let $Q^*_{XY}$ be the minimizer in~\eqref{eq:ann_FA},
so $\lambda(Q_{XY}^*,R)\geq\tau$ and
$\EFA^{\mbox{\tiny a}}(\tau,R)=D_{\mbox{\tiny m}}(Q^*_Y)+[I_{Q^*}(X;Y)-R]_+$.
Fix $\epsilon>0$. The proof has two parts: the lower bound, 
asserting that 
$$\liminf_{n\to\infty}\frac{1}{n}\log\alpha_n(\tau,R,\calC)\ge
-\EFA^{\mbox{\tiny a}}(\tau,R)~~a.s.,$$ 
and the upper bound, namely,
$$\limsup_{n\to\infty}\frac{1}{n}\log\alpha_n(\tau,R,\calC)\le -\EFA^{\mbox{\tiny a}}(\tau,R)~~a.s.$$

Beginning from the lower bound,
we distinguish between two cases, according to whether $Q_{XY}^*$ is sparse or bulk.
Suppose first that $Q_{XY}^*$ is sparse ($I_{Q^*}(X;Y)>R$). Then,
in view of eq.\ \eqref{eq:ann_FA}, we have
\begin{equation}
\EFA^{\mbox{\tiny a}}(\tau,R)=D_{\mbox{\tiny m}}(Q^*_Y)+I_{Q^*}(X;Y)-R.
\end{equation}
Define 
\begin{equation}
\ell(Q_{XY})\dfn H_Q(Y|X)+D_{\mbox{\tiny c}}(Q_{XY}),
\end{equation}
so that $W^n(y^n|x^n)=e^{-n\ell(Q_{XY})}$ for $(x^n,y^n)\in\calT(Q_{XY})$.
Every codeword $x^n(m)$ together with every
$y^n\in\mathcal{A}_m\dfn\calT(Q^*_{Y|X}|x^n(m))$
contributes $W^n(y^n|x^n(m))=e^{-n\ell(Q^*_{XY})}$ to $S(y^n,\calC)$.
Since $P_Y^{\otimes n}(y^n)=e^{-n[H_{Q^*}(Y)+D_{\mbox{\tiny m}}(Q^*_Y)]}$ and
$M=e^{nR}$, we get
\begin{align}
\Lambda(y^n,\calC)
&\geq\frac{1}{n}\log\frac{e^{-n\ell(Q^*_{XY})}}{e^{nR}
\cdot e^{-n(H_{Q^*}(Y)+D_{\mbox{\tiny m}}(Q^*_Y))}}
\notag\\
&= D_{\mbox{\tiny m}}(Q^*_Y)-D_{\mbox{\tiny c}}(Q^*_{XY})+I_{Q^*}(X;Y)-R
\notag\\
 &= \lambda(Q^*_{XY},R)
\geq\tau.
\label{eq:LLR_sparse}
\end{align}
Here we used the fact that for
sparse $Q_{XY}^*$, $I_{Q^*}(X;Y)-R=[I_{Q^*}(X;Y)-R]_+$, and so,
$D_{\mbox{\tiny m}}(Q_Y^*)-D_{\mbox{\tiny c}}(Q_{XY}^*)+I_{Q^*}(X;Y)-R=\lambda(Q^*_{XY},R)$.
It follows that every $y^n\in\mathcal{A}_m$ satisfies the FA condition,
for every $m$ and every codebook $\calC$.
The probability of $\mathcal{A}_m$ under $P_Y^{\otimes n}$ satisfies,
at the exponential scale,
\begin{equation}
P_Y^{\otimes n}(\mathcal{A}_m)
\doteq e^{nH_{Q^*}(Y|X)}\cdot e^{-n[H_{Q^*}(Y)+D_{\mbox{\tiny m}}(Q^*_Y)]}
=e^{-n[I_{Q^*}(X;Y)+D_{\mbox{\tiny m}}(Q^*_Y)]}.
\label{eq:Am_measure}
\end{equation}
Clearly, 
\begin{equation}
\alpha_n(\tau,R,\calC)\ge P_Y^{\otimes n}\left\{\bigcup_{m=1}^M
\mathcal{A}_m\right\},
\end{equation}
which we lower-bound via inclusion-exclusion.
This is the key structural observation: $P_Y^{\otimes n}(\mathcal{A}_m)\doteq e^{-n[I_{Q^*}(X;Y)+D_{\mbox{\tiny m}}(Q^*_Y)]}$
is the \emph{same deterministic quantity} $Z$ for every $m$ and every codebook $\calC$,
since it depends only on the fixed type $Q^*_{XY}$ and the fixed distribution $P_Y$,
not on which specific codeword $x^n(m)$ was drawn.
Accordingly,
\begin{equation}
\sum_{m=1}^M P_Y^{\otimes n}(\mathcal{A}_m)
= M\cdot Z
\doteq e^{-n\EFA^{\mbox{\tiny a}}(\tau,R)}
\label{eq:firstmoment}
\end{equation}
deterministically.
For the pairwise overlaps, since $x^n(m)$ and $x^n(m')$ are independent for $m\ne m'$,
for each fixed $y^n\in\calT(Q^*_Y)$ the events $\{y^n\in\mathcal{A}_m\}$ and
$\{y^n\in\mathcal{A}_{m'}\}$ are independent. Hence:
\begin{eqnarray}
\E\left\{P_Y^{\otimes n}(\mathcal{A}_m\cap\mathcal{A}_{m'})\right\}
&=&\sum_{y^n}P_Y^{\otimes n}(y^n)\cdot
\Pr\{\mathcal{A}_m~\mbox{includes}~y^n\}\cdot\Pr\{\mathcal{A}_{m'}~\mbox{includes}~y^n\}\nonumber\\
&\doteq& e^{-n[D_{\mbox{\tiny m}}(Q^*_Y)+2I_{Q^*}(X;Y)]},
\end{eqnarray}
and so:
\begin{equation}
\E\left\{\sum_{m\ne m'}
P_Y^{\otimes n}(\mathcal{A}_m\cap\mathcal{A}_{m'})\right\}
\doteq e^{-n[D_{\mbox{\tiny m}}(Q^*_Y)+2I_{Q^*}(X;Y)-2R]}.
\label{eq:overlap_expectation}
\end{equation}
In the sparse case $I_{Q^*}(X;Y)>R$, the exponent in
\eqref{eq:overlap_expectation} exceeds 
$\EFA^{\mbox{\tiny a}}(\tau,R)=D_{\mbox{\tiny m}}(Q^*_Y)+I_{Q^*}(X;Y)-R$
by the amount $I_{Q^*}(X;Y)-R>0$.
Hence by Markov's inequality and the Borel--Cantelli lemma, for almost every $\calC$
and all large $n$:
\begin{equation}
\sum_{m\ne m'}P_Y^{\otimes n}(\mathcal{A}_m\cap\mathcal{A}_{m'})
\le\tfrac{1}{2}e^{-n[\EFA^{\mbox{\tiny a}}(\tau,R)+\epsilon/2]}.
\label{eq:overlap_as}
\end{equation}
By the inclusion-exclusion principle, using~\eqref{eq:firstmoment}
and the a.s.\ event~\eqref{eq:overlap_as}:
\begin{equation}
	P_Y^{\otimes n}\left\{\bigcup_{m=1}^M\mathcal{A}_m\right\}
	\ge\sum_{m=1}^M P_Y^{\otimes n}(\mathcal{A}_m)
  -\sum_{m\ne m'}P_Y^{\otimes n}(\mathcal{A}_m\cap\mathcal{A}_{m'})
\ge\tfrac{1}{2}M\cdot Z
\doteq\tfrac{1}{2}e^{-n\EFA^{\mbox{\tiny a}}(\tau,R)},
\end{equation}
so $\alpha_n(\tau,R,\calC)\ge e^{-n[\EFA^{\mbox{\tiny a}}(\tau,R)+\epsilon]}$ a.s.

Consider next the case where $Q_{XY}^*$ is bulk ($I_{Q^*}(X;Y)\leq R$). In this case,
\begin{equation}
\EFA^{\mbox{\tiny a}}(\tau,R)=D_{\mbox{\tiny m}}(Q^*_Y).
\end{equation}
Since the feasible set $\{\lambda(Q_{XY},R)\geq\tau\}$ is closed and the
infimum of $D_{\mbox{\tiny m}}(Q_Y)+[I_Q(X;Y)-R]_+$ over it equals $\EFA^{\mbox{\tiny a}}(\tau,R)$,
for any $\epsilon>0$ we may choose $Q^*_{XY}$ with $\lambda(Q^*_{XY},R)>\tau$
(strictly inside the feasible set)
and $D_{\mbox{\tiny m}}(Q^*_Y)\leq\EFA^{\mbox{\tiny a}}(\tau,R)+\epsilon/2$.
As mentioned before, $N(Q^*_{XY}|y^n)$ is a binomial random variable with $M=e^{nR}$ trials and
success rate of the exponential order of $e^{-nI_{Q^*}(X;Y)}$.
Since $\mu_n\dfn\E\{N(Q^*_{XY}|y^n)\}\doteq e^{n[R-I_{Q^*}(X;Y)]}\to\infty$, 
the Chernoff bound yields, for any $\delta\in(0,1)$:
\begin{equation}
  \Pr\bigl\{N(Q^*_{XY}|y^n)<(1-\delta)\mu_n\bigr\}
  \leq e^{-\delta^2\mu_n/2}
	= \exp\{-\delta^2 e^{n[R-I_{Q^*}(X;Y)]}/2\},
\end{equation}
which is doubly exponentially small as a function of $n$.
By the union bound over $|\calT(Q^*_Y)|\leq e^{nH_{Q^*}(Y)}$
$Q_Y^*$-typical sequences and the Borel-Cantelli lemma,
almost surely, for all large $n$,
$N(Q^*_{XY}|y^n)\geq(1-\delta)\mu_n$ for all $y^n\in\calT(Q^*_Y)$ at the same
time. On this a.s.\ event, for every $y^n\in\calT(Q^*_Y)$:
\begin{align}
\Lambda(y^n,\calC)
&\geq\frac{1}{n}\log\frac{(1-\delta)\,e^{n[R-I_{Q^*}(X;Y)]}
\cdot e^{-n\ell(Q^*_{XY})}}{e^{nR}
\cdot e^{-n[H_{Q^*}(Y)+D_{\mbox{\tiny m}}(Q^*_Y)]}}
\notag\\
&= D_{\mbox{\tiny m}}(Q^*_Y)-D_{\mbox{\tiny c}}(Q^*_{XY})+\frac{\log(1-\delta)}{n}
\notag\\
&= \lambda(Q^*_{XY},R)+\frac{\log(1-\delta)}{n}
> \tau+\frac{\log(1-\delta)}{n}
\label{eq:LLR_bulk}
\end{align}
which exceeds $\tau$ for all large $n$
(since $\lambda(Q^*_{XY},R)>\tau$ and $\frac{\log(1-\delta)}{n}\to 0$),
where we have used $\lambda(Q_{XY}^*,R)=D_{\mbox{\tiny
m}}(Q_Y^*)-D_{\mbox{\tiny c}}(Q_{XY}^*)$ (as $[I_{Q^*}(X;Y)-R]_+=0$ for bulk $Q^*$).
Therefore, $\alpha_n(\tau,R,\calC)\geq P_Y^{\otimes n}(\calT(Q^*_Y))
\doteq e^{-nD_{\mbox{\tiny m}}(Q^*_Y)} \geq e^{-n[\EFA^{\mbox{\tiny a}}(\tau,R)+\epsilon/2]}$,
so $\alpha_n(\tau,R,\calC)\geq e^{-n[\EFA^{\mbox{\tiny a}}(\tau,R)+\epsilon]}$
a.s.\ for all large $n$.

Moving on to the upper bound,
we use the fact that $\alpha_n(\tau,R,\calC)$ is a non-negative random
variable whose expectation over the randomness of $\calC$ is
the \emph{annealed} FA probability:
\begin{equation}
  \E\{\alpha_n(\tau,R,\calC)\}
  = \sum_{y^n} P_Y^{\otimes n}(y^n)
    \cdot \Pr\bigl\{\Lambda(y^n,\calC)\geq\tau\bigr\}
  = \alpha_n^{\mbox{\tiny a}}(\tau,R).
  \label{eq:annealed_FA}
\end{equation}
By the annealed analysis of~\cite[Theorem 1]{MerhavCompanion},
\begin{equation}
  \alpha_n^{\mbox{\tiny a}}(\tau,R)
  \doteq e^{-n\EFA^{\mbox{\tiny a}}(\tau,R)},
\end{equation}
so for any $\epsilon>0$ and all large $n$, we have by Markov's inequality,
\begin{equation}
\Pr\left\{\calC:~\alpha_n(\tau,R,\calC)\ge 
e^{-n[\EFA^{\mbox{\tiny
a}}(\tau,R)-\epsilon]}\right\}\le\frac{e^{-n\EFA^{\mbox{\tiny
a}}(\tau,R)}}{e^{-n[\EFA^{\mbox{\tiny a}}(\tau,R)-\epsilon]}}
=e^{-n\epsilon},
\end{equation}
which is summable across all $n\ge 1$, and so, by the Borel-Cantelli lemma,
\begin{equation}
  \alpha_n(\tau,R,\calC) \leq e^{-n[\EFA^{\mbox{\tiny a}}(\tau,R)-\epsilon]}
  \label{eq:alpha_ub_total}
\end{equation}
eventually almost surely for all large $n$.
Combining with the a.s.\ lower bound, we observe that
$\frac{1}{n}\log\alpha_n(\tau,R,\calC)\to-\EFA^{\mbox{\tiny a}}(\tau,R)$ a.s.
\end{proof}

The result for the expected logarithm of the FA probability is asserted in the
following corollary:

\begin{corollary}
\label{cor:FA}
\begin{equation}
  -\frac{1}{n}\E\bigl\{\log\alpha_n(\tau,R,\calC)\bigr\}
  \;\xrightarrow{n\to\infty}\;\EFA^{\mbox{\tiny a}}(\tau,R),
\end{equation}
or, equivalently,
\begin{equation}
\EFA^{\mbox{\tiny q}}(\tau,R)=\EFA^{\mbox{\tiny a}}(\tau,R).
\end{equation}
\end{corollary}

\begin{proof}
Let $U_n\dfn-\frac{1}{n}\log\alpha_n(\tau,R,\calC)\geq 0$.
By Theorem~\ref{thm:FA}, we know that $U_n\to\EFA^{\mbox{\tiny a}}(\tau,R)$ a.s.
By Fatou's lemma,
\begin{equation}
\liminf_{n\to\infty}\E\{U_n\}\geq\E\{\liminf_{n\to\infty}U_n\}=\EFA^{\mbox{\tiny a}}(\tau,R).
\end{equation}
On the other hand, by Jensen's inequality applied to the logarithmic function,
\begin{equation}
\E\{U_n\}
=-\frac{1}{n}\E\{\log\alpha_n(\tau,R,\calC)\}
\leq -\frac{1}{n}\log\E\{\alpha_n(\tau,R,\calC)\}
=-\frac{1}{n}\log\alpha_n^{\mbox{\tiny a}}(\tau,R)
\to\EFA^{\mbox{\tiny a}}(\tau,R).
\end{equation}
Combining both bounds, we conclude that $\lim_{n\to\infty}\E\{U_n\}=\EFA^{\mbox{\tiny a}}(\tau,R)$.
\end{proof}

\section{The MD Exponent}
\label{sec:MD}

We now present parallel results for the MD exponent.

\begin{theorem}
\label{thm:MD}
For all $\tau\in\mathbb{R}$ and $R>0$:
\begin{equation}
\frac{1}{n}\log\beta_n(\tau,R,\calC)
\;\xrightarrow{n\to\infty}\;-\EMD^{\mbox{\tiny a}}(\tau,R)
\quad\textup{a.s.\ over }\calC.
\label{eq:MD_as}
\end{equation}
\end{theorem}

\begin{proof}
The upper bound follows by Markov's inequality and the Borel--Cantelli lemma,
exactly as for the FA exponent. We henceforth focus on the lower bound.

By the symmetry of the random coding mechanism, 
it may be assumed without loss of generality that $x^n(1)$ is the transmitted codeword.
Let $x^n\dfn x^n(1)$ be drawn uniformly from $\calT(P_X)$ independently of
$\calC'\dfn\{x^n(2),\ldots,x^n(M)\}$,
and let $y^n\sim W^n(\cdot|x^n)$.
Define
\begin{equation}
S'(y^n,\calC')\dfn\sum_{m=2}^M W^n(y^n|x^n(m)),
\end{equation}
so that 
\begin{equation}
S(y^n,\calC)=W^n(y^n|x^n)+S'(y^n,\calC').
\end{equation}
All almost-sure statements below are over the random selection of $\calC$.
Pick a feasible type $Q^*_{XY}$ with $Q^*_X=P_X$,
$\lambda(Q^*_{XY},R)<\tau$, $\Delta(Q^*_Y,R)<\tau$, and
$D_{\mbox{\tiny c}}(Q^*_{XY})\leq\EMD^{\mbox{\tiny a}}(\tau,R)+\epsilon/2$.
Such a type exists since the feasible set is non-empty with infimum
$\EMD^{\mbox{\tiny a}}(\tau,R)$.
Set $\delta\dfn\tau-\Delta(Q^*_Y,R)>0$.

For any $y^n\in\calT(Q^*_Y)$, the LLR $\Lambda(y^n,\calC)$ simplifies to:
\begin{equation}
  \Lambda(y^n,\calC)
  =\frac{1}{n}\log S(y^n,\calC)-R+H_{Q^*}(Y)+D_{\mbox{\tiny m}}(Q^*_Y).
\end{equation}
Define
\begin{equation}
  \theta(Q^*_Y) \dfn \tau+R-H_{Q^*}(Y)-D_{\mbox{\tiny m}}(Q^*_Y),
\end{equation}
so that $\Lambda(y^n,\calC)<\tau$ iff $S(y^n,\calC)<e^{n\theta(Q^*_Y)}$.
For $(x^n,y^n)\in\calT(Q^*_{XY})$,
$W^n(y^n|x^n)=e^{-n\ell(Q^*_{XY})}$
where 
\begin{equation}
\ell(Q^*_{XY})\dfn H_{Q^*}(Y|X)+D_{\mbox{\tiny c}}(Q^*_{XY}).
\end{equation}
Using the identity $H_{Q^*}(Y)=H_{Q^*}(Y|X)+I_{Q^*}(X;Y)$, we have
\begin{equation}
\ell(Q^*_{XY})+\theta(Q^*_Y)
=\tau-\lambda(Q^*_{XY},R)+[R-I_{Q^*}(X;Y)]_+>0
\label{eq:theta_plus_ell}
\end{equation}
since $\lambda(Q^*_{XY},R)<\tau$ and $[R-I_{Q^*}(X;Y)]_+\geq 0$.
Hence $W^n(y^n|x^n)<e^{n\theta(Q^*_Y)}$, and for all large $n$:
\begin{equation}
W^n(y^n|x^n) < \tfrac{1}{2}e^{n\theta(Q^*_Y)}.
\label{eq:Wn_lt_theta}
\end{equation}
Next, let us decompose $S'(y^n,\calC')$ by joint types:
\begin{equation}
  S'(y^n,\calC')
	=\sum_{\substack{Q_{XY}':\,Q'_X=P_X \\ Q'_Y=Q^*_Y}}
    N'(Q'_{XY}|y^n)\cdot e^{-n\ell(Q'_{XY})},
\end{equation}
where $N'(Q'_{XY}|y^n)\dfn\#\{m\geq 2:(x^n(m),y^n)\in\calT(Q'_{XY})\}$.
Since $y^n\in\calT(Q^*_Y)$,
only types $Q_{XY}'$ with $Q'_Y=Q^*_Y$ contribute to $S'(y^n,\calC')$.

We next distinguish between sparse and bulk types $Q_{XY}'$.

For sparse types,
$\E\{N'(Q'_{XY}|y^n)\}\doteq e^{-n[I_{Q'}(X;Y)-R]}\to 0$,
so by Markov's inequality and the Borel--Cantelli lemma,
$N'(Q'_{XY}|y^n)=0$ eventually a.s.\ for all large $n$,
simultaneously over all $y^n\in\calT(Q^*_Y)$
and all (polynomially many) sparse $Q_{XY}'$.

For bulk types, 
$N'(Q'_{XY}|y^n)$ is a binomial random variable with $M-1$ trials
and success rate of the exponential order of $e^{-nI_{Q'}(X;Y)}$. Denoting $\mu_n'\dfn\E\{N'(Q'_{XY}|y^n)\}\doteq
e^{n[R-I_{Q'}(X;Y)]}$, the
Chernoff bound yields $\Pr\{N'(Q'_{XY}|y^n)>2\mu_n'\}\leq e^{-\mu_n'/3}$,
which is doubly exponentially small.
Applying the union bound over all $y^n\in\calT(Q^*_Y)$
and all (polynomially many) bulk types $Q_{XY}'$,
the total failure probability is still doubly exponentially small.
By the Borel--Cantelli lemma, eventually almost surely for all large $n$:
\begin{equation}
N'(Q'_{XY}|y^n)\leq 2\mu_n'
\quad\text{for all }y^n\in\calT(Q^*_Y)\text{ and all bulk }Q'.
  \label{eq:Omega_LB}
\end{equation}
Call this a.s.\ event $\Omega$.
On $\Omega$, for any $y^n\in\calT(Q^*_Y)$, we bound $S'(y^n,\calC')$
by summing only over bulk types (sparse types contribute nothing on $\Omega$):
\begin{align}
S'(y^n,\calC')
	&\leq\sum_{\substack{Q_{XY}':\,Q'_Y=Q^*_Y \\ I_{Q'}(X;Y)\leq R}}
2\mu_n'\cdot e^{-n\ell(Q'_{XY})}
\notag\\
	&=2\sum_{Q_{XY}'} e^{n[R-I_{Q'}(X;Y)]}\cdot
  e^{-n[H_{Q'}(Y|X)+D_{\mbox{\tiny c}}(Q'_{XY})]}
\notag\\
	&=2\sum_{Q_{XY}'} e^{n[R-H_{Q^*}(Y)-D_{\mbox{\tiny c}}(Q'_{XY})]}
\label{eq:Sprime_step}\\
	&=2\sum_{Q_{XY}'} e^{n[\theta(Q^*_Y)+\lambda(Q'_{XY},R)-\tau]}
\label{eq:Sprime_identity}\\
&\leq 2\,\mathrm{poly}(n)\cdot e^{n\theta(Q^*_Y)}\cdot e^{-n\delta}
\notag\\
&<\tfrac{1}{2}e^{n\theta(Q^*_Y)}
\label{eq:Sprime_ub}
\end{align}
for all large $n$, where $\mathrm{poly}(n)$ denotes the (polynomial) number of types.
Step~\eqref{eq:Sprime_step} uses $I_{Q'}(X;Y)+H_{Q'}(Y|X)=H_{Q'}(Y)=H_{Q^*}(Y)$
since $Q'_Y=Q^*_Y$.
Step~\eqref{eq:Sprime_identity} uses the identity,
valid for all bulk $Q_{XY}'$ with $Q'_Y=Q^*_Y$ (so $[I_{Q'}(X;Y)-R]_+=0$):
\begin{align}
R-H_{Q^*}(Y)-D_{\mbox{\tiny c}}(Q'_{XY})
&=\theta(Q^*_Y)+\lambda(Q'_{XY},R)-\tau,
\label{eq:key_identity}
\end{align}
which follows by substituting
$\theta(Q^*_Y)=\tau+R-H_{Q^*}(Y)-D_{\mbox{\tiny m}}(Q^*_Y)$
and $\lambda(Q'_{XY},R)=D_{\mbox{\tiny m}}(Q^*_Y)-D_{\mbox{\tiny c}}(Q'_{XY})$
(the latter is because $[I_{Q'}(X;Y)-R]_+=0$ for bulk $Q_{XY}'$).
The penultimate step uses
$\lambda(Q'_{XY},R)\leq\Delta(Q^*_Y,R)=\tau-\delta$
for every bulk $Q_{XY}'$ with $Q'_Y=Q^*_Y$, by definition of $\Delta(Q^*_Y,R)$.

It follows that on
$\Omega$, for all large $n$
and all $y^n\in\calT(Q^*_{Y|X}|x^n)$:
\begin{equation}
S(y^n,\calC)
= W^n(y^n|x^n)+S'(y^n,\calC')
<\tfrac{1}{2}e^{n\theta(Q^*_Y)}+\tfrac{1}{2}e^{n\theta(Q^*_Y)}
=e^{n\theta(Q^*_Y)},
\end{equation}
and so, $\Lambda(y^n,\calC)<\tau$ for every such $y^n$.
Therefore:
\begin{eqnarray}
\beta_n(\tau,R,\calC)
&\geq& W^n\!\bigl(\calT(Q^*_{Y|X}|x^n)\bigm|x^n\bigr)
\notag\\
&\doteq& e^{-nD_{\mbox{\tiny c}}(Q^*_{XY})}
\notag\\
&\geq& e^{-n[\EMD^{\mbox{\tiny a}}(\tau,R)+\epsilon]}
\quad\text{a.s.,}
\end{eqnarray}
where the second line uses the method of types
and the third uses $D_{\mbox{\tiny c}}(Q^*_{XY})\leq\EMD^{\mbox{\tiny a}}(\tau,R)+\epsilon/2$.
Combining both bounds and letting $\epsilon\to 0$:
$\frac{1}{n}\log\beta_n(\tau,R,\calC)\to-\EMD^{\mbox{\tiny a}}(\tau,R)$ a.s.
\end{proof}

We have the following analogous corollary for the quenched MD exponent.

\begin{corollary}
\label{cor:MD}
\begin{equation}
  -\frac{1}{n}\E\bigl\{\log\beta_n(\tau,R,\calC)\bigr\}
  \;\xrightarrow{n\to\infty}\;\EMD^{\mbox{\tiny a}}(\tau,R),
\end{equation}
or, equivalently,
\begin{equation}
\EMD^{\mbox{\tiny q}}(\tau,R)=
\EMD^{\mbox{\tiny a}}(\tau,R).
\end{equation}
\end{corollary}

\begin{proof}
The proof is the same as for Corollary~\ref{cor:FA},
with $U_n\dfn-\frac{1}{n}\log\beta_n(\tau,R,\calC)$ and
$\EFA^{\mbox{\tiny a}}(\tau,R)$ replaced by $\EMD^{\mbox{\tiny a}}(\tau,R)$.
\end{proof}

\begin{remark}
Since Corollaries~\ref{cor:FA} and~\ref{cor:MD} establish that the annealed
and quenched exponents coincide, there is no longer any need to distinguish
between them. Henceforth, we write $\EFA(\tau,R)$ and $\EMD(\tau,R)$
for the common value, dropping the superscript $^{\rm a}$.
\end{remark}

\section{Discussion on the Annealed-Quenched Equivalence}
\label{sec:jensen}

The self-averaging results of Sections~\ref{sec:FA} and \ref{sec:MD}
tell us that for large $n$, both $\frac{1}{n}\log\alpha_n(\tau,R,\calC)$ and
$\frac{1}{n}\log\beta_n(\tau,R,\calC)$
are essentially constants across the vast majority of codebooks $\calC$:
the \emph{Jensen gap} of both exponents vanishes.
Such a behavior should not be taken for granted, because it is not common to other exponents
associated with random coding, for example, the classical (annealed) random coding error exponent
(see, e.g., \cite{Gallager65}, \cite[Section 5.6]{Gallager68}) 
as opposed to the typical-code (quenched) error exponent, which differ in general at
low coding rates --- see, e.g., \cite{BargForney02}, \cite{merhavtypical} and
references therein. Moreover, the Jensen gap between the annealed and quenched
free energies of disordered systems with random parameters is positive in general at low
temperatures. It is then natural to wonder what it is that
makes the difference between the problem setting at hand and the other 
problem settings in that respect.

The key structural reason is that $S(y^n,\calC)=\sum_{m=1}^M W^n(y^n|x^n(m))$
is a sum of $M=e^{nR}$ \emph{independent} terms for each fixed $y^n$
(independence follows from the independent generation of codewords).
The coefficient of variation therefore decays as $e^{-nR/2}\to 0$,
so $S(y^n,\calC)/\E\{S(y^n,\calC)\}\to 1$ for each $y^n$,
and the self-averaging of $\frac{1}{n}\log\alpha_n(\tau,R,\calC)$ follows.
This is in contrast to spin glasses, where the terms of the partition function
are \emph{correlated} through the shared coupling matrix,
causing large relative fluctuations and a positive Jensen gap;
and to channel coding, where the error probability depends on a
\emph{maximum} rather than a sum, making it sensitive to rare bad codewords.
For $\beta_n(\tau,R,\calC)$, the mechanism is slightly different:
the dominant contribution is the deterministic transmitted-type weight
$W^n(y^n|x^n(1))=e^{-n\ell(Q^*_{XY})}$, and the interferer sum concentrating
\emph{below} the threshold $e^{n\theta(Q^*_Y)}$ is what gives self-averaging.

\section{Mismatched Detection}
\label{sec:mismatch}

So far we have assumed that the detector knows the true channel $W$
and computes the optimal Neyman--Pearson LLR.
We now study what happens when the detector is
\emph{mismatched}: it assumes a wrong channel model $\tilde{W}:\calX\to\calY$
in place of $W$.
This is the practically relevant scenario where the
channel is estimated imperfectly or chosen for computational
convenience rather than accuracy.
The mismatched setting also opens a genuinely new theoretical question:
does self-averaging survive, and how do the exponents degrade?

Let $\tilde{W}:\calX\to\calY$ be the detector's assumed channel, with
output marginal $\tilde{P}_Y(y)\dfn\sum_x P_X(x)\tilde{W}(y|x)$.
The codebook $\calC$ is still drawn from $\calT(P_X)$
under the true channel $W$, and the true hypotheses are as before:
\begin{align*}
\calH_0&:~~y^n\sim P_Y^{\otimes n},\\
\calH_1&:~~y^n\sim P_{Y^n|\calC}
\quad\text{(true mixture output under }W).
\end{align*}
The detector constructs its LLR
assuming the channel $\tilde{W}$:
\begin{equation}
\tilde{\Lambda}(y^n,\calC)
\dfn \frac{1}{n}\log\frac{\sum_{m=1}^M \tilde{W}^n(y^n|x^n(m))}
    {M\cdot \tilde{P}_Y^{\otimes n}(y^n)}.
  \label{eq:mismatch_LLR}
\end{equation}
The mismatched detector decides in favor of $\calH_1$ iff
$\tilde{\Lambda}(y^n,\calC)\geq\tau$, otherwise it decides in favor of
$\calH_0$.
The error probabilities are:
\begin{align}
\tilde{\alpha}_n(\tau,R,\calC)
&\dfn \sum_{\{y^n:~\tilde{\Lambda}(y^n,\calC)\geq\tau\}} P_Y^{\otimes n}(y^n)
\label{eq:alpha_mm}\\
\tilde{\beta}_n(\tau,R,\calC)
&\dfn\frac{1}{M}\sum_{m=1}^M\sum_{\{y^n:~\tilde{\Lambda}(y^n,\calC)<\tau\}}
    W^n(y^n|x^n(m)).
  \label{eq:beta_mm}
\end{align}
Note that $\tilde{\alpha}_n$ and $\tilde{\beta}_n$ are always evaluated under
the true measures $P_Y^{\otimes n}$ and $W^n$, as error probabilities are
defined with respect to the true distribution regardless of the detector's beliefs.
What differs between modeling choices is the hypotheses the detector formulates.
In the model above, the true hypotheses are $\calH_0$ and $\calH_1$ as stated,
and the mismatch lies entirely in the test statistic $\tilde{\Lambda}$
(e.g.\ due to a wrong channel estimate or a computationally convenient approximation).
A detector who \emph{genuinely believes} the channel is $\tilde{W}$ would instead
formulate hypotheses
$\tilde{\calH}_0: y^n\sim\tilde{P}_Y^{\otimes n}$ and
$\tilde{\calH}_1: y^n\sim\frac{1}{M}\sum_m\tilde{W}^n(\cdot|x^n(m))$,
giving the LRT statistic $\tilde{\Lambda}(y^n,\calC)$ a different operational meaning
even though its mathematical form is identical.

For a joint type $Q_{XY}$ with $Q_X=P_X$, define the
mismatched analogues of the quantities in Section~\ref{sec:setup}:
\begin{align}
\tilde{D}_{\mbox{\tiny m}}(Q_Y)
&\dfn D(Q_Y\|\tilde{P}_Y),
\\
\tilde{D}_{\mbox{\tiny c}}(Q_{XY})
&\dfn D(Q_{Y|X}\|\tilde{W}|P_X),
\\
\tilde{\ell}(Q_{XY})
&\dfn H_Q(Y|X)+\tilde{D}_{\mbox{\tiny c}}(Q_{XY}),
\\
\tilde{\lambda}(Q_{XY},R)
&\dfn \tilde{D}_{\mbox{\tiny m}}(Q_Y)-\tilde{D}_{\mbox{\tiny c}}(Q_{XY})+[I_Q(X;Y)-R]_+,
\\
\tilde{\theta}(Q_Y)
&\dfn \tau+R-H_Q(Y)-\tilde{D}_{\mbox{\tiny m}}(Q_Y),
\\
\tilde{\Delta}(Q_Y,R)
&\dfn \max_{\substack{Q'_{XY}:\,Q'_X=P_X,\,Q'_Y=Q_Y\\
I_{Q'}(X;Y)\leq R}}
\tilde{\lambda}(Q'_{XY},R).
\end{align}
Clearly, for $(x^n,y^n)\in\calT(Q_{XY})$,
$\tilde{W}^n(y^n|x^n)=e^{-n\tilde{\ell}(Q_{XY})}$ and
$\tilde{P}_Y^{\otimes n}(y^n)=e^{-n[H_Q(Y)+\tilde{D}_{\mbox{\tiny m}}(Q_Y)]}$,
so the condition $\tilde{\Lambda}(y^n,\calC)<\tau$ is equivalent to
$\tilde{S}(y^n,\calC)<e^{n\tilde{\theta}(Q_Y)}$,
where $\tilde{S}(y^n,\calC)\dfn\sum_m \tilde{W}^n(y^n|x^n(m))$.

The annealed exponents follow from the same derivations
as in~\cite{MerhavCompanion}, replacing $W$ with $\tilde{W}$ in the
test statistic while keeping the true underlying probability measures for the error probabilities.

\begin{theorem}
\label{thm:mm_annealed}
\begin{align}
\tilde{E}_{\rm FA}(\tau,R)
&=\min_{\substack{Q_{XY}:\,Q_X=P_X\\
\tilde{\lambda}(Q_{XY},R)\geq\tau}}
\bigl[D_{\mbox{\tiny m}}(Q_Y)+[I_Q(X;Y)-R]_+\bigr],
\label{eq:EFA_mm}\\
\tilde{E}_{\rm MD}(\tau,R)
&=\min_{\substack{Q_{XY}:\,Q_X=P_X\\
\tilde{\lambda}(Q_{XY},R)<\tau\\
\tilde{\Delta}(Q_Y,R)<\tau}}
D_{\mbox{\tiny c}}(Q_{XY}).
\label{eq:EMD_mm}
\end{align}
\end{theorem}

\begin{proof}
	The proof follows~\cite[Theorem~1]{MerhavCompanion} 
with the following substitutions throughout:
$W\to\tilde{W}$, $P_Y\to\tilde{P}_Y$,
$\ell(Q_{XY})\to\tilde{\ell}(Q_{XY})$, $\lambda(Q_{XY},R)\to\tilde{\lambda}(Q_{XY},R)$,
$\theta(Q_Y)\to\tilde{\theta}(Q_Y)$, $\Delta(Q_Y,R)\to\tilde{\Delta}(Q_Y,R)$,
and $S(y^n,\calC)\to\tilde{S}(y^n,\calC)
\dfn\sum_m\tilde{W}^n(y^n|x^n(m))$.
The key steps are as follows.

The annealed FA probability is:
\begin{equation*}
\E\{\tilde{\alpha}_n(\tau,R,\calC)\}
=\sum_{y^n}P_Y^{\otimes n}(y^n)\cdot
  \Pr\!\bigl\{\tilde{S}(y^n,\calC)\geq
e^{n\tilde{\theta}(\hat{P}_{y^n})}\bigr\}.
\end{equation*}
For $y^n\in\calT(Q_Y)$:
$P_Y^{\otimes n}(y^n)=e^{-n[H_Q(Y)+D_{\mbox{\tiny m}}(Q_Y)]}$;
$\tilde{W}^n(y^n|x^n(m))=e^{-n\tilde{\ell}(Q_{XY})}$
for $(x^n(m),y^n)\in\calT(Q_{XY})$.
Crucially, the type-class enumerator
$N(Q_{XY}|y^n)\sim\mathrm{Bin}(M,e^{-nI_Q(X;Y)})$
has the same distribution as in the matched case, as
it counts codewords of a given joint type with $y^n$,
which depends only on the codebook, not on $\tilde{W}$.
Applying Theorem~4.1 of~\cite{MW25} to bound
$\Pr\{\tilde{S}(y^n,\calC)\geq e^{n\tilde{\theta}(\hat{P}_{y^n})}\}$
exactly as in~\cite{MerhavCompanion},
the dominant type minimizes the true cost
$D_{\mbox{\tiny m}}(Q_Y)+[I_Q(X;Y)-R]_+$
subject to the mismatched constraint $\tilde{\lambda}(Q_{XY},R)\geq\tau$,
yielding~\eqref{eq:EFA_mm}.

The annealed MD probability is:
\begin{equation*}
\E\{\tilde{\beta}_n(\tau,R,\calC)\}
=\frac{1}{M}\sum_{m=1}^M\sum_{y^n}W^n(y^n|x^n(m))\cdot
\Pr\bigl\{
\tilde{S}'(y^n,\calC')<e^{n\tilde{\theta}(\hat{P}_{y^n})}
\bigr\}.
\end{equation*}
The probability that $(x^n,y^n)$ has joint type $Q_{XY}$
is $\doteq e^{-nD_{\mbox{\tiny c}}(Q_{XY})}$.
The event $\{\tilde{S}(y^n,\calC)<e^{n\tilde{\theta}(\hat{P}_{y^n})}\}$
is analyzed via Theorem~4.3 of~\cite{MW25}
applied to the counts $N'(Q'_{XY}|y^n)$,
with thresholds computed using $\tilde{\ell}(Q_{XY})$ and $\tilde{\theta}(Q_Y)$.
The feasibility condition reduces to
$\tilde{\lambda}(Q_{XY},R)<\tau$ and $\tilde{\Delta}(Q_Y,R)<\tau$,
and the dominant type minimizes the true $D_{\mbox{\tiny c}}(Q_{XY})$
over this mismatched feasible set,
yielding~\eqref{eq:EMD_mm}.
\end{proof}

\noindent
In the FA exponent~\eqref{eq:EFA_mm}:
the \emph{constraint} $\tilde{\lambda}(Q_{XY},R)\geq\tau$
uses the mismatched $\tilde{\lambda}(Q_{XY},R)$ (the test statistic decides when to fire),
while the \emph{cost} $D_{\mbox{\tiny m}}(Q_Y)+[I_Q(X;Y)-R]_+$
uses the true divergence $D_{\mbox{\tiny m}}(Q_Y)$ (the true FA probability weight).
In the MD exponent~\eqref{eq:EMD_mm}:
the constraints use $\tilde{\lambda}(Q_{XY},R)$ and $\tilde{\Delta}(Q_Y,R)$,
while the cost $D_{\mbox{\tiny c}}(Q_{XY})$ uses the true channel.
The detector's beliefs ($\tilde{W}$) control \emph{when} the test fires,
but the true channel ($W$) controls \emph{how costly} each error is.

The next theorem establishes self-averaging under mismatch;
the mismatched counterparts of Corollaries~\ref{cor:FA} and~\ref{cor:MD}
then follow by the same Fatou--Jensen argument.

\begin{theorem}
\label{thm:mm_sa}
For all $\tau\in\mathbb{R}$, $R>0$, and almost every $\calC$:
\begin{align}
\frac{1}{n}\log\tilde{\alpha}_n(\tau,R,\calC)
&\;\xrightarrow{n\to\infty}\; -\tilde{E}_{\rm FA}(\tau,R),
\\
\frac{1}{n}\log\tilde{\beta}_n(\tau,R,\calC)
&\;\xrightarrow{n\to\infty}\; -\tilde{E}_{\rm MD}(\tau,R).
\end{align}
\end{theorem}

\begin{proof}
The proofs of Theorems~\ref{thm:FA} and~\ref{thm:MD}
depend on the codebook only through:
(a) the independence of codewords,
(b) the type-class structure ($W^n(y^n|x^n)=e^{-n\ell(Q_{XY})}$
for $(x^n,y^n)\in\calT(Q_{XY})$), and
(c) linearity of expectation over the codebook.
None of these depend on whether the test uses $W$ or $\tilde{W}$.
Replacing $W$ by $\tilde{W}$ throughout
(i.e.\ $S(y^n,\calC)\to\tilde{S}(y^n,\calC)$, $\ell(Q_{XY})\to\tilde{\ell}(Q_{XY})$, $\lambda(Q_{XY},R)\to\tilde{\lambda}(Q_{XY},R)$,
$\theta(Q_Y)\to\tilde{\theta}(Q_Y)$, $\Delta(Q_Y,R)\to\tilde{\Delta}(Q_Y,R)$)
yields the mismatched annealed exponents as the limits,
while all concentration arguments go through unchanged.
\end{proof}

Self-averaging is therefore a structural property of
the random coding mechanism, robust to detector mismatch.
However, while the Jensen gap remains zero,
the exponent \emph{values} change with $\tilde{W}$ and can degrade severely.

The more interesting part of this section concerns the phase structure and the
mismatched critical rate.
The central question is the following: for which rates $R$ does a tradeoff interval
between the two positive exponents
exist under the mismatched detector?
Define:
\begin{equation}
\tau^*(R)\dfn\tilde{\lambda}(P_X\otimes W,R)
=D(P_Y\|\tilde{P}_Y)-D(W\|\tilde{W}|P_X)+[I(X;Y)-R]_+,
\label{eq:tau_star_mm}
\end{equation}
the value of $\tilde{\lambda}(Q_{XY},R)$ at the true channel type $Q_{XY}^*=P_X\otimes W$.
In particular, at $R=I(X;Y)$:
\begin{equation}
\tau^*\dfn\tau^*(I(X;Y))=D(P_Y\|\tilde{P}_Y)-D(W\|\tilde{W}|P_X)\leq 0,
\label{eq:tau_star}
\end{equation}
with $\tau^*=0$ iff $\tilde{W}=W$ (by the data processing inequality).
By the same inequality, $\tau^*(R)\leq[I(X;Y)-R]_+$ for all $R$,
with equality in particular when $\tilde{W}=W$.

\begin{proposition}
\label{prop:mm_phase}
\begin{enumerate}[label=(\roman*)]
\item The mismatched critical rate, the infimum of rates above which
no tradeoff interval exists, is:
\begin{equation}
R_{\mbox{\tiny c}} = \inf_{\substack{Q_{XY}:\,Q_X=P_X \\ Q_Y=P_Y}}
\max\bigl\{I_Q(X;Y),\;
I(X;Y)-D(W\|\tilde{W}|P_X)+\tilde{D}_{\mbox{\tiny c}}(Q_{XY})\bigr\}.
\label{eq:Rmm_exact}
\end{equation}
In the matched case, $R_{\mbox{\tiny c}}=I(X;Y)$,
which is attained by $Q_{XY}=P_X\otimes W$.

\item For $R<R_{\mbox{\tiny c}}$, the tradeoff interval
$(\tilde{\Delta}(P_Y,R),\tau^*(R))$ is non-empty (this is precisely the content
of part~(i)), and both exponents are simultaneously positive for
$\tau\in(\tilde{\Delta}(P_Y,R),\tau^*(R))$.
\end{enumerate}
\end{proposition}

\begin{proof}
Let $Q_{XY}^*\dfn P_X\otimes W$, so that $D_{\mbox{\tiny m}}(Q_{Y}^*)=0$,
$D_{\mbox{\tiny c}}(Q_{XY}^*)=0$, and $\tilde{\lambda}(Q_{XY}^*,R)=\tau^*(R)$.
The tradeoff interval $(\tilde{\Delta}(P_Y,R),\tau^*(R))$ is non-empty iff
$\tilde{\Delta}(P_Y,R)<\tau^*(R)$.
For bulk $Q_{XY}'$ with $Q'_Y=P_Y$ and $I_{Q'}(X;Y)\leq R$:
\begin{equation}
\tilde{\lambda}(Q'_{XY},R)
=D(P_Y\|\tilde{P}_Y) - \tilde{D}_{\mbox{\tiny c}}(Q'_{XY}),
\end{equation}
so $\tilde{\Delta}(P_Y,R)=D(P_Y\|\tilde{P}_Y)
-\min_{\{Q_{XY}':Q'_Y=P_Y,I_{Q'}(X;Y)\leq R\}}\tilde{D}_{\mbox{\tiny
c}}(Q_{XY}')$.
And $\tau^*(R)=D(P_Y\|\tilde{P}_Y)-D(W\|\tilde{W}|P_X)+(I(X;Y)-R)$
(using $[I(X;Y)-R]_+=I(X;Y)-R$ since $R<I(X;Y)$ in this case).
The condition $\tilde{\Delta}(P_Y,R)<\tau^*(R)$ becomes the following: for all $Q_{XY}'$ with
$Q'_Y=P_Y$ and $I_{Q'}(X;Y)\leq R$,
$R < I(X;Y)-D(W\|\tilde{W}|P_X)+\tilde{D}_{\mbox{\tiny c}}(Q_{XY}')$.
Combined with $R\geq I_{Q'}(X;Y)$, the tradeoff interval first disappears
when there exists $Q_{XY}'$ with $Q'_Y=P_Y$ such that
\begin{equation}
  R \geq \max\bigl\{I_{Q'}(X;Y),\;
    I(X;Y)-D(W\|\tilde{W}|P_X)+\tilde{D}_{\mbox{\tiny c}}(Q'_{XY})\bigr\}.
\end{equation}
Taking the infimum over all such $Q_{XY}'$ gives~\eqref{eq:Rmm_exact}.
For $\tilde{W}=W$: $Q_{XY}'=Q_{XY}^*$ gives $I(X;Y)$,
and any other $Q_{XY}'$ with $Q'_Y=P_Y$ has $D_{\mbox{\tiny c}}(Q_{XY}')>0$, giving a larger value.
Part (ii) follows from the zero-region characterizations
of $\tilde{E}_{\rm FA}(\tau,R)$ and $\tilde{E}_{\rm MD}(\tau,R)$.
\end{proof}

\subsection{Discussion}

Taking $Q_{XY}'=Q_{XY}^*$ in~\eqref{eq:Rmm_exact} gives $R_{\mbox{\tiny c}}\leq I(X;Y)$.
The data processing inequality gives $\tilde{D}_{\mbox{\tiny c}}(Q_{XY})\geq D(P_Y\|\tilde{P}_Y)$
for all $Q_{XY}$ with $Q_Y=P_Y$, so from~\eqref{eq:Rmm_exact}:
\begin{equation}
R_{\mbox{\tiny c}}
\geq I(X;Y)-D(W\|\tilde{W}|P_X)+D(P_Y\|\tilde{P}_Y)
\dfn I_{\rm GMI}^{(1)},
\label{eq:critical_rate_bounds}
\end{equation}
where $I_{\rm GMI}^{(1)}$ is the generalized mutual information (GMI)
of mismatched decoding evaluated at
$s=1$~\cite{KaplanShamai93,ScarlettGuillen14}, and
its connection to the full GMI is discussed below.
Thus, $I_{\rm GMI}^{(1)}\leq R_{\mbox{\tiny c}}\leq I(X;Y)$ (the same interval
that contains the LM capacity in mismatched channel decoding),
and equivalently,~\eqref{eq:Rmm_exact} can be written as:
\begin{equation}
  R_{\mbox{\tiny c}} = \inf_{\substack{Q_{XY}:\,Q_X=P_X \\ Q_Y=P_Y}}
    \max\bigl\{I_{Q}(X;Y),\;
      I_{\rm GMI}^{(1)} + \tilde{D}_{\mbox{\tiny c}}(Q_{XY}) - D(P_Y\|\tilde{P}_Y)\bigr\}.
  \label{eq:Rmm_GMI}
\end{equation}

The matched soft-covering point $(\tau=0,\,R=I(X;Y))$
shifts to $(\tau^*,\,R_{\mbox{\tiny c}})$ under mismatch:
the threshold shifts from $0$ to $\tau^*\leq 0$,
and the critical rate shifts from $I(X;Y)$ to $R_{\mbox{\tiny c}}\leq I(X;Y)$.
The threshold shift $|\tau^*|=I(X;Y)-I_{\rm GMI}^{(1)}$
equals the GMI capacity loss (see below).
Mismatch can only lower $R_{\mbox{\tiny c}}$, never raise it above $I(X;Y)$,
because for $R\geq I(X;Y)$ the true-type codewords
($N(Q^*_{XY}|y^n)\doteq e^{n(R-I(X;Y))}\geq 1$, where $Q^*_{XY}=P_X\otimes W$)
push $\tilde{S}(y^n,\calC)$ past the threshold on their own,
regardless of the weights $e^{-n\tilde{\ell}(Q_{XY})}$ of other types.

It is interesting to examine two extreme examples of mismatch.
The first is a degenerate mismatched channel $\tilde{W}(y|x)=V(y)$,
which gives $\tilde{\Lambda}(y^n,\calC)\equiv 0$ identically
(the mismatched sum $\tilde{S}(y^n,\calC)=M\cdot V^{\otimes n}(y^n)$
cancels the denominator exactly).
For $\tau>0$: $\tilde{\alpha}_n=0$ and $\tilde{\beta}_n=1$;
for $\tau<0$: $\tilde{\alpha}_n=1$ and $\tilde{\beta}_n=0$.
In neither case are both exponents simultaneously positive,
so $R_{\mbox{\tiny c}}=0$: the collapse occurs at every $R>0$,
the worst possible case.
The second example is a BSC with the wrong polarity,
where $W=\mathrm{BSC}(p)$,
$\tilde{W}=\mathrm{BSC}(1-p)$, and $P_X=(\frac{1}{2},\frac{1}{2})$.
Since $\tilde{W}(y|x)=W(y|\bar{x})$, flipping every codeword bit maps the
mismatched sum to the matched sum: $\tilde{S}(y^n,\calC)\overset{d}{=}S(y^n,\calC)$
over balanced binary codebooks, because $\overline{\calC}\overset{d}{=}\calC$
under the uniform input distribution.
Consequently the mismatched and matched exponents are identical,
and $R_{\mbox{\tiny c}}=I(X;Y)$,
despite $D(W\|\tilde{W}|P_X)$ being large and $I_{\rm GMI}^{(1)}\ll I(X;Y)$.
The test statistic carries no less information about $y^n$ under $\tilde{W}$
than under $W$; the wrong-polarity assumption is invisible to the LRT.
This shows that the lower bound $R_{\mbox{\tiny c}}\geq I_{\rm GMI}^{(1)}$
can be far from tight:
codebook symmetry can make a severely mismatched detector
perform exactly as well as the matched one.

The quantity $I_{\rm GMI}^{(1)}$ defined in~\eqref{eq:critical_rate_bounds}
can be written as:
\begin{equation}
  I_{\rm GMI}^{(1)}
  = \sum_{x,y}P_X(x)W(y|x)\log\frac{\tilde{W}(y|x)}{\tilde{P}_Y(y)},
  \label{eq:GMI1}
\end{equation}
i.e.\ the mutual information computed with $\tilde{W}$ in the numerator
and $\tilde{P}_Y$ in the denominator, averaged under the true $P_X\otimes W$.
This is the $s=1$ special case of the full GMI~\cite{KaplanShamai93}:
\begin{equation}
  I_{\rm GMI}
  = \max_{s\geq 0}\sum_{x,y}P_X(x)W(y|x)
    \log\frac{\tilde{W}(y|x)^s}{\sum_{x'}P_X(x')\tilde{W}(y|x')^s},
  \label{eq:GMI}
\end{equation}
which satisfies $I_{\rm GMI}\geq I_{\rm GMI}^{(1)}$ and $I_{\rm GMI}\geq 0$ always,
while $I_{\rm GMI}^{(1)}$ can be negative for severely mismatched $\tilde{W}$.
Note that $I_{\rm GMI}^{(1)}\geq 0$ iff $D(W\|\tilde{W}|P_X)\leq I(X;Y)+D(P_Y\|\tilde{P}_Y)$,
i.e.\ iff the mismatch is not too severe.
The lower bound $R_{\mbox{\tiny c}}\geq I_{\rm GMI}^{(1)}$ in~\eqref{eq:critical_rate_bounds}
is therefore informative only when $I_{\rm GMI}^{(1)}>0$;
when $I_{\rm GMI}^{(1)}\leq 0$ the bound is trivial since $R_{\mbox{\tiny c}}\geq 0$ always.
The identity $I_{\rm GMI}^{(1)}=I(X;Y)-D(W\|\tilde{W}|P_X)+D(P_Y\|\tilde{P}_Y)$
follows from~\eqref{eq:GMI1} by writing
$$\log\frac{\tilde{W}(y|x)}{\tilde{P}_Y(y)}
=\log\frac{W(y|x)}{P_Y(y)}-\log\frac{W(y|x)}{\tilde{W}(y|x)}+\log\frac{P_Y(y)}{\tilde{P}_Y(y)}$$
and taking expectations under $P_X\otimes W$.
Therefore $|\tau^*|=D(W\|\tilde{W}|P_X)-D(P_Y\|\tilde{P}_Y)=I(X;Y)-I_{\rm GMI}^{(1)}$:
the soft-covering threshold shift equals the $s=1$ GMI capacity loss.

\section{Conclusion}
\label{sec:conclusion}

We have proved that both the FA and MD error exponents
for the soft-covering hypothesis testing problem are
\emph{self-averaging}:
for almost every random codebook $\calC$,
both $\frac{1}{n}\log\alpha_n(\calC)\to-\EFA(\tau,R)$
and $\frac{1}{n}\log\beta_n(\calC)\to-\EMD(\tau,R)$.
The proofs use elementary tools: Markov's inequality
(for both upper bounds), the Chernoff bound and
Borel--Cantelli (for the lower bounds).
The results confirm that the annealed exponents of~\cite{MerhavCompanion}
describe the typical codebook behavior,
not just the ensemble average.

We then extended the analysis to mismatched detection,
where the detector uses a wrong channel $\tilde{W}$.
Self-averaging persists under mismatch,
and the mismatched exponents are characterized by the same
minimization formulas as the matched case,
with $\tilde{\lambda}(Q_{XY},R)$, $\tilde{\Delta}(Q_Y,R)$ replacing their matched counterparts
in the constraints while the true divergences $D_{\mbox{\tiny m}}(Q_Y)$, $D_{\mbox{\tiny c}}(Q_{XY})$
appear in the costs.
The key structural finding is that for $R\geq I(X;Y)$,
no tradeoff interval exists regardless of $\tilde{W}$,
but for $R<R_{\mbox{\tiny c}}$ the tradeoff interval exists
with width depending on the severity of the mismatch.
The mismatched critical rate satisfies
$I_{\rm GMI}^{(1)}\leq R_{\mbox{\tiny c}}\leq I(X;Y)$,
and the natural operating threshold shifts from $\tau=0$ to
$\tau^*=D(P_Y\|\tilde{P}_Y)-D(W\|\tilde{W}|P_X)\leq 0$.
The magnitude of this shift equals the mismatched decoding capacity loss:
$|\tau^*|=I(X;Y)-I_{\rm GMI}^{(1)}$.

\bibliographystyle{IEEEtran}

\end{document}